\documentclass[twoside,12pt]{article}
\usepackage{amsmath,amssymb,amsthm,amsfonts,mathrsfs,wasysym,latexsym,times,lineno, subfigure,color,booktabs}
\usepackage{amsmath,amsfonts}
\usepackage{multirow}
\usepackage{array}
\usepackage{amsthm}
\usepackage{graphicx}
\usepackage{color}
\usepackage{multicol}
\usepackage{float}
\usepackage{cite}

\newtheorem{thm}{Theorem}[section]

\newtheorem{definition}{Definition}[section]

\newtheorem{lemma}[thm]{Lemma}

\newtheorem{theorem}[thm]{Theorem}

\title{Solution independence and self-referential instances}
\author{
  Guangyan Zhou$^1$,\quad Bin Wang$^2$,\quad Jianxin Wang$^3$,\quad Ke Xu$^{4*}$
}
\date{}

\begin{document}
\maketitle
\begin{center}
\footnotesize
\noindent
$^1$Department of Mathematics and Statistics, Beijing Technology and Business University, Beijing, 100048, China\\
zhouguangyan@btbu.edu.cn\\[2pt]
$^2$Academy of Mathematics and Systems Science, Chinese Academy of Sciences, Beijing, 100080, China\\
wangbin@amss.ac.cn\\[2pt]
$^3$School of Computer Science and Engineering, Central South University, Changsha, 410083, China\\
jxwang@mail.csu.edu.cn\\[2pt]
$^4$State Key Lab of Complex and Critical Software Environment, Beihang University, Beijing, 100083, China\\
kexu@buaa.edu.cn\\[4pt]

\end{center}

\begin{abstract}
In this paper, we investigate the hitting set problem and demonstrate that solution independence is the crucial property underlying the construction of self-referential instances. As a special case of the hitting set problem, the vertex cover problem lacks the solution independence property. This distinction accounts for its ability to evade exhaustive search, as correlations among candidate solutions can be leveraged to compress the overall search space. In contrast, the dominating set problem on hypergraphs, which is also a special case of the hitting set problem, satisfies the solution independence property, thereby enabling the construction of self-referential instances. Moreover, we prove that these self-referential instances possess an irreducible property, implying that any algorithm for solving such instances must process nearly the entire graph to yield a correct solution.

\end{abstract}
\footnotetext[1]{Corresponding author. }

\section{Introduction}
Self-reference stands as a foundational and indispensable concept across mathematics and computer science, serving as a powerful tool for exploring the limits of computation and formal logic. Its central idea is to construct objects or statements that refer to themselves, thereby enabling a system to reflect on its own expressive power. This inward reflection often reveals intrinsic limitations or even contradictions, making self-reference particularly effective in establishing impossibility results.

The intuition behind self-reference can be traced back to classical semantic paradoxes, such as the Liar Paradox (``this statement is false”), where self-reference leads to logical inconsistency. Rather than indicating a failure of logic, such paradoxes highlight structural phenomena that arise when a system is sufficiently expressive to describe itself.

In 1931, Gödel \cite{godel1931} formalized self-reference through his encoding technique known as Gödel numbering. By transforming syntactic statements into arithmetic objects, he constructed a statement that asserts its own unprovability, leading to the incompleteness theorems. These results fundamentally established the inherent limitations of finite formal systems.
Building upon these foundational ideas, Church \cite{church1936} and Turing \cite{turing1936} extended the notion of self-reference and pioneered computability theory. In particular, the \emph{undecidability} proof of the halting problem (i.e., the impossibility for a general algorithm to determine if any given program will halt or run forever) relies heavily on a self-referential construction. A hypothetical halting-detection algorithm is forced to reason about its own behavior, leading to a logical contradiction that proves the non-existence of such an algorithm. This argument establishes the inherent impossibility of algorithmically solving the halting problem, thereby precisely characterizing the decidability boundary of algorithms. Using a self-referential construction analogous to that employed by Turing, Hartmanis and Stearns \cite{HS1965} demonstrated in their pioneering work on computational complexity that more time allows for solving a broader class of computational problems.

Recently, self-reference has found novel applications in characterizing and proving extreme hardness (i.e., the inherent necessity of exhaustive search). Following the line of using self-reference to establish impossibility results, Xu and Zhou \cite{xu2025}, Li et al. \cite{Li2025}, and Zhou \cite{zhou2026} constructed self-referential instances for the Constraint Satisfaction Problem, the Clique problem, and the Dominating Set problem, respectively.
Such instances form an infinite set whose negation under symmetric mappings is equivalent to the set itself. This is analogous to Gödel’s self-referential statement whose unprovability is equivalent to the statement itself.
Just as the truth value of the self-referential statement is indistinguishable within finite formal systems, the solvability of the self-referential instances cannot be distinguished by non-exhaustive algorithms. This fundamental form of \emph{indistinguishability} stems from the inherent gap between syntax and semantics, i.e., the distinction between the part and the whole.

The reason why self-referential instances can be constructed for the above problems lies in their common property: the near-independence of candidate solutions. Specifically, for any two randomly selected candidate solutions, the probability that both constitute valid solutions is nearly equal to the product of their individual probabilities of being solutions. For simplicity, in this paper we refer to this near-independence of candidate solutions simply as \emph{solution independence}.
For many NP-complete problems, including 3-SAT, 3-Coloring, 0-1 Knapsack, Hamiltonian Cycle, and Vertex Cover, the property of solution independence fails to hold (i.e., their candidate solutions are mutually correlated). Consequently, these problems are able to evade exhaustive search by compressing the overall search space of candidate solutions. For example, the Vertex Cover problem admits algorithms that are significantly more efficient than naive exhaustive search \cite{chen2001}.
The intrinsic difference between NP-complete problems (such as 3-SAT) and P problems (such as 2-SAT) lies in the varying degrees of correlation among their candidate solutions. Specifically, 2-SAT exhibits strong correlations, enabling efficient algorithms, whereas 3-SAT exhibits relatively weaker correlations, resulting in a substantially larger search space. Consequently, although both problems evade naive exhaustive search, 2-SAT can be solved within a far smaller search space than 3-SAT.

Although the intrinsic difference between 2-SAT and 3-SAT is intuitively clear, it remains rather difficult to carry out a quantitative analysis of the search space size based solely on the strength of correlations. This is precisely one key reason why proving P$\ne$ NP by analyzing the differences between 2-SAT and 3-SAT is extremely challenging. In contrast, the perspective of solution independence enables us not only to intuitively identify the source of computational hardness, but also to provide a rigorous proof. To better illustrate this idea, we begin with a simple yet illuminating example.

Consider the following \emph{Coin and Box Problem}. Suppose that there are $ n $ boxes arranged in a row, each containing a coin that is either heads (positive) or tails (negative). The task is to determine whether at least one box is positive, and to analyze how many boxes must be examined in the worst case. The answer depends crucially on the  correlation structure among the box states.
If strong correlations exist, for example, if a box being negative implies that the adjacent box on its right is also negative, then examining only the leftmost box suffices.
If weaker correlations exist, such as symmetry between the left and right halves, then inspecting only half of the boxes is sufficient. However, when the box states are completely independent, the problem becomes fundamentally different. Even after examining the first $ n-1 $ boxes and finding them negative, the state of the $ n $-th box remains entirely undetermined, that is, it could be either positive or negative, and flipping its state does not affect the states of the first $n-1$ boxes. This allows the construction of self-referential instances by flipping the state of the $n$-th box: the instance where all boxes are negative, and the instance where only the last box is positive. These two instances can be transformed into each other by a single flip, and inspecting only the first $n-1$ boxes cannot distinguish them. Therefore, in the worst case, any subproblem consisting of $n-1$ boxes is insufficient, and inspection of all $n$ boxes is required to ensure a correct answer.

The above example demonstrates that the number of required inspections is determined entirely by the degree of correlation among the box states. The condition of independence requires an exhaustive box-by-box inspection. Constructing self-referential instances provides an effective approach to proving the inherent necessity of exhaustive search through proof by contradiction. Specifically, the existence of such instances enables straightforward construction of counter-instances from original instances, thereby facilitating proofs by contradiction. This is precisely the missing piece in contemporary computational complexity theory.

It is worth noting that in their foundational work on parameterized complexity theory, Downey and Fellows  \cite{downey1999} identified a crucial observation: NP-complete problems exhibit intrinsic differences in computational hardness, and accordingly established a corresponding hierarchy of complexity classes.
For example, the Vertex Cover problem is FPT, whereas the Clique problem and the Dominating Set problem are W[1]-hard and W[2]-hard, respectively. Furthermore, several researchers have studied the computational hardness of some classical problems based on parameterized complexity theory \cite{chen2006} or the strong exponential time hypothesis \cite{will2015,cygan2016}, using reductions to explain why these problems cannot avoid exhaustive search. In this paper, we take the Hitting Set problem as an example and, starting from the perspective of solution independence, provide a comparative analysis to explain why exhaustive search can be avoided in some cases but is unavoidable in others.

\section{Two special cases of the hitting set problem}

Given a set $U$ of $n$ elements, and a collection  $\Sigma=\{F_1,...,F_m\}$ of subsets, the \textbf{hitting set problem} is to find a subset $S\subset U$ which intersects every $F_i\in\Sigma$. In the following, we examine two special cases of this problem, the vertex cover problem and dominating set problem. We show that while the vertex cover problem does not exhibit solution independence, the dominating set problem does. Moreover, no sublinear-sized induced subgraph can capture the full combinatorial structure of its solutions.
\subsection{The vertex cover problem: Absence of solution independence}
Let $G=G(n,p)$ be a random graph. A vertex cover of $G$ is a subset of vertices such that every edge $(u, v)$ of $G$ has at least one endpoint in this subset. This is a special case of the hitting set problem with $|F_i|=2$. We show that two random solutions of this problem exhibit positive correlations, thereby violating solution independence. Let $S_1, S_2$ be two random subsets of size $k$ with overlap $|S_1 \cap S_2| = i$. Then
\begin{align*}
\Pr(S_1 \text{ is a vertex cover})
&= (1-p)^{\binom{n-k}{2}},\\
\Pr(S_1, S_2 \text{ are vertex covers})
&= (1-p)^{2\binom{n-k}{2} - \binom{n-2k+i}{2}}.
\end{align*}
To quantify dependence, consider the ratio
\[
\frac{\Pr(S_1, S_2 \text{ are vertex covers})}{\Pr(S_1 \text{ is a vertex cover}) \Pr(S_2 \text{ is a vertex cover})}
= (1-p)^{-\binom{n-2k+i}{2}}.
\]

Now evaluate this ratio under different regimes. For $p = \frac{c}{n}$, the typical size of the minimum vertex cover satisfies $k = \Theta(n)$; specifically, $k \approx \frac{c}{2}n$ for small $c$, and $k$ increases monotonically  with $c$. For $p = \Theta(1)$, it is well known that $k = n - \Theta(\log n)$.

In all these regimes, for overlaps $i$ such that $n - 2k + i = \Theta(n)$, we have
\[
(1-p)^{-\binom{n-2k+i}{2}} = \exp\big(\Theta(n^2 p)\big),
\]
which is exponentially large in $n$. This correlation structure extends naturally to the vertex cover problem on $d$-uniform random hypergraphs. In this case,
\begin{align*}
\frac{\Pr(S_1, S_2 \text{ are vertex covers})}{\Pr(S_1 \text{ is a vertex cover}) \Pr(S_2 \text{ is a vertex cover})}
= (1-p)^{-\binom{n-2k+i}{d}},
\end{align*}
which is $\exp\bigl(\Theta(n^d p)\bigr)$ for certain overlaps.

The above analysis reveals a positive correlation between solutions with a large overlap. This violates the near-independence condition typically required for the second moment method. Consequently, these correlations constitute a fundamental obstacle to applying second moment techniques in the vertex cover problem.

\subsection{The dominating set problem: Solution independence and irreducibility}
Let $V$ be a vertex set with $|V| = n$, and let $d \geq 2$ be an integer. Consider the random $d$-uniform hypergraph $G = G_d(n,p)$, whose hyperedge set $E \subseteq \binom{V}{d}\triangleq \{ S \subseteq V : |S| = d \}$ is formed by including each $d$-subset independently with probability $p \in (0,1)$.

A subset $S \subseteq V$ is called a \textbf{(weak) dominating set} of $G$ if for every vertex $v \in V$, either $v \in S$ or there exists a vertex $u \in S$ such that $u$ and $v$ are contained in a common hyperedge.

To facilitate our analysis, we reformulate the problem as a hitting set problem. Let $\mathcal{S} \subseteq \binom{V}{d}$ denote the random family of hyperedges. For each vertex $u \in V$, define

\[
S_u = \{u\} \cup \left\{ v \in V \setminus \{u\} : \exists \ S \in \mathcal{S},\; \{u,v\} \subseteq S \right\}.
\]
That is, $S_u$ consists of $u$ together with all vertices that share at least one hyperedge with $u$.
A subset $H \subseteq V$ is called a \textbf{hitting set} if for every $u \in V$,
\[
H \cap S_u \neq \varnothing.
\]

This establishes a one-to-one correspondence between weak dominating sets of $G$ and hitting sets of the family $\{S_u\}_{u \in V}$. 

A special case of $d=2$ was studied in \cite{zhou2026}, where the notion of reducibility capturing whether the problem can be confined to a sublinear-sized subinstance was studied. By applying the second moment, it was shown that, in contrast to the vertex cover problem, when $p=p(n)$ tends to $1-e^{-1}$, there exist dominating sets of size
$k=\ln n$, and two candidate $k$-dominating sets exhibit near-independence. In\cite{zhou2026}, it was shown that the random graph $G_2(n,p)$ is irreducible with high probability.

In this paper, we focuse on the case $d \geq 3$. A key difference from the case $d=2$ is that the edge probability scales as $\Theta(n^{-(d-2)})$. Despite this difference, the two settings share similar structural properties, and thus the overall proof strategy we adopt is largely analogous.

\subsection{Main results}


\begin{theorem}\label{th:main}
For $d \geq 3$, the dominating set problem on the random hypergraph $G_d(n,p)$ is irreducible with high probability.
\end{theorem}

Theorem~\ref{th:main} shows that, under the notion of reducibility, no sublinear-sized induced subgraph can capture the full combinatorial structure of dominating sets in $G_d(n,p)$. This implies that local subgraphs cannot faithfully represent global graph properties. Equivalently, any algorithm solving this problem must process nearly the entire graph in the worst case to obtain a correct answer.
From the perspective of the hitting set formulation, this irreducibility stems from the near-independence among the candidate dominating sets of size $k=\ln n$.

\section{Proof of Theorem  \ref{th:main}}

In this section, we show that with positive probability, either there exists a unique dominating set of size $k$, or there exists a quasi-dominating set of size $k$ that leaves exactly one vertex undominated. If one  inspects only a sublinear-sized subgraph,  a symmetric mapping can be applied to the residual subgraph which interchanges instances where a dominating set exists with those where it does not. As a result, no sublinear-sized induced subgraph can reliably distinguish between these two cases. Consequently, the dominating set problem on the random hypergraph $G_d(n,p)$ is irreducible with high probability.

Before proceeding, we highlight an important structural property that, in contrast to the vertex cover problem, dominating sets exhibit an independence property. Specifically, let $S_1$ and $S_2$ be two random subsets of vertices of size $k$ with overlap $|S_1\cap S_2|=i$. Then
\begin{align*}
\frac{\Pr(S_1, S_2 \text{ are dominating sets})}{\Pr(S_1 \text{ is a dominating set}) \Pr(S_2 \text{ is a dominating set})}
&=\left[\frac{1-2(1-p)^{M}+(1-p)^{M_i}}{(1-(1-p)^{M})^2}\right]^{n-2k+i},
\end{align*}
where $M$ and $M_i$ are combinatorial parameters defined later. As will be shown in the subsequent analysis, this ratio simplifies to 
\[\exp\left\{\frac{(\ln^2n)^{2-\frac ik}}{n^{1-\frac ik}}\right\}=1+o(1),\]
indicating that  two random dominating sets are asymptotically independent.

\subsection{Existence of a unique dominating set of size $k$}
In this section we show that, with positive probability, the random graph $G_d(n,p)$ contains a unique dominating sets of size $k=\ln n$. In the following, we denote $f\approx g$, if $f=(1+o(1)g$.

\begin{lemma}
    Let $X$ be the number of dominating sets of size $k$ in $G_d(n,p)$, then
\begin{align*}
\mathbf{E}[X]&=\binom{n}{k}(1-(1-p)^{M})^{n-k},\\
    \mathbf{E}[X^2]&=\sum_{i=0}^k\binom{n}{k}\binom{k}{i}\binom{n-k}{k-i}(1-(1-p)^{M})^{2(k-i)}(1-2(1-p)^{M}+(1-p)^{M_i})^{n-2k+i},
\end{align*}
where 
\begin{align*}
M&=\binom{n-1}{d-1}-\binom{n-1-k}{d-1}\approx\frac{kn^{d-2}}{(d-2)!},\\
M_i&=\binom{n-1}{d-1}-\binom{n-1-(2k-i)}{d-1}\approx\frac{(2k-i)n^{d-2}}{(d-2)!}.
\end{align*}  
\end{lemma}

\begin{proof}
Let $S$ be a candidate dominating set of size $k$. There are $\binom nk$ such sets $S$. For a vertex $v\in V\backslash S$, if $v$ is not dominated by $S$, then every hyperedge containing $v$ contains no vertex from $S$. There are totally $\binom{n-1}{d-1}$ hyperedges containing $v$, and among them, $\binom{n-k-1}{d-1}$ hyperedges contain $v$ but avoid $S$. Hence the probability that $v$ is not dominated by $S$ is
$$\Pr(v \text{ is not dominated by }S)=(1-p)^M,$$
where $M=\binom{n-1}{d-1}-\binom{n-k-1}{d-1}$.

To compute the second moment $\mathbf{E}[X^2]$, let $S$ and $S'$ be two dominating sets of size $k=\ln n$, and suppose $|S\cap S'|=i$ with $0\le i\le k$. Then,  $$|S\cup S'|=|S'\cup S'|=k-i,\quad |V\backslash\{S\cup S'\}|=n-2k+i.$$

First, the number of ordered pairs $(S,S')$ with $|S|=|S'|=k$ and $|S\cap S'|=i$ is
$$\binom{n}{k}\binom{k}{i}\binom{n-k}{k-i}.$$  Second, vertices in $S'\backslash S$ must be dominated by $S$, and this probability  is $$(1-(1-p)^M)^{k-i}.$$
Similarly, vertices in $S\backslash S'$ should be dominated by $S'$, which occurs with the same probability. Moreover, any vertex $v\in\backslash\{S\cup S'\}$ must be dominated by both $S$ and $S'$. Let $E_1,E_2$ be the events that $v$ is dominated by $S,S'$, respectively.
We have
 \begin{align*}
 \Pr(v\text{ is dominated by both }S_1\text{ and }S_2) =&\Pr(E_1\cap E_2)\\=&1-\Pr(\overline{E_1})-\Pr(\overline{E_2})+\Pr( \overline{E_1}\cap\overline{E_2}).
 \end{align*}

By symmetry, 
$$\Pr(\overline{E_1})=\Pr(\overline{E_2})=(1-p)^M.$$

If $v$ is not dominated by either $S$ or $S'$, then every hyperedge containing $v$ contains no vertex from $S\cup S'$. There are $\binom{n-1}{d-1}$ hyperedges containing $v$ in total, and among them, $\binom{n-1-(2k-i)}{d-1}$ hyperedges contain $v$ but avoid $S\cup S'$, thus
$$\Pr( \overline{E_1}\cap\overline{E_2})=(1-p)^{M_i},$$
where $M_i=\binom{n-1}{d-1}-\binom{n-1-(2k-i)}{d-1}$.
Therefore
$$\Pr(S,S'\text{ are both dominating sets})=(1-(1-p)^{M})^{2(k-i)}(1-2(1-p)^{M}+(1-p)^{M_i})^{n-2k+i}.$$

\end{proof}
In the following, we tacitly choose the hyperedge probability $p=p(n)$ such that 
\begin{align}\label{eq:delta}
    E[X]=\delta+o(1)
\end{align}
for some constant $0<\delta<1$.  To obtain the asymptotic value of $p$, we take logarithms on both sides of (\ref{eq:delta}) and obtain 
$$\ln\binom nk+(n-k)\ln(1-(1-p)^{M})=\ln\delta.$$
Note that $\ln\binom nk\approx\ln^2n-\ln n\ln\ln n$, then $\ln(1-(1-p)^{M})\approx-\frac{\ln^2 n}{n}$. Thus
\begin{align}
(1-p)^{M}\approx1-\exp\left\{-\frac{\ln^2n}n\right\}\approx\frac{\ln^2n}n.
\end{align}
We now have $\ln(1-p)\approx-\frac{\ln n}{M}\approx-\frac{(d-2)!}{n^{d-2}}$. Therefore,
\begin{align}
p\approx1-\exp\left\{-\frac{(d-2)!}{n^{d-2}}\right\}\approx\frac{(d-2)!}{n^{d-2}}.
\end{align}

\begin{lemma}\label{lem:onethird}
In  $G_d(n,p)$,
$$\frac{\delta}{\delta+1}\le \mathbf{Pr}(X>0)\le\delta.$$
\end{lemma}

\begin{proof}
The upper bound follows immediately from Markov’s inequality:
\begin{equation}
\mathbf{Pr}(X>0)\le \mathbf{E}[X]=\delta+o(1).
\end{equation}
For a lower bound we apply the second moment method.  The quantity $\mathbf{E}[X^2]$ counts ordered pairs of dominating sets of size $k$. Let $F(i)$ denote the contribution from pairs whose intersection has size $i$, then
\begin{align*}
\mathbf{E}[X^2]&=\sum_{i=0}^k F(i),
\end{align*}
where
 \begin{align*}
F(i):=\binom{n}{k}\binom{k}{i}\binom{n-k}{k-i}(1-(1-p)^{M})^{2(k-i)}(1-2(1-p)^{M}+(1-p)^{M_0})^{n-2k+i}.
\end{align*}

 If $i=0$, we have
  \begin{align*}
F(0)=\binom{n}{k}\binom{n-k}{k}\big(1-(1-p)^M\big)^{2n-2k}=(1+o(1))\mathbf{E}[X]^2.
\end{align*}
If $i=k$, then 
  \begin{align*}
F(k)=\mathbf{E}[X].
\end{align*}
For $1\le i\le k-1$, we have

  \begin{align*}
\frac{F(i)}{\mathbf{E}^2[X]}=\frac{\binom ki\binom{n-k}{k-i}}{\binom nk}\left[\frac{(1-2(1-p)^M+(1-p)^{M_i})}{(1-(1-p)^M)^2}\right]^{(n-2k+i)}.
\end{align*}

Note that  $k=\ln n$, we apply the following asymptotic estimates for $1\le i\le k-1$: 
$$\frac{\binom ki\binom{n-k}{k-i}}{\binom nk}\le\frac{k^{2i}}{n}, \quad (1-p)^M\approx\frac{\ln^2 n}n, \quad (1-p)^{M_i}\approx\left(\frac{\ln^2n}{n}\right)^{2-\frac ik}.$$ This gives 

\begin{align*}
\frac{F(i)}{\mathbf{E}[X]}
&\le (1+o(1))k^{2i}n^{-i}\exp\left\{\frac{(\ln^2n)^{2-\frac ik}}{n^{1-\frac ik}}\right\}.
\end{align*}

Summing over $i$, we obtain 
\begin{equation}\label{eq:cauchy}
\frac{\mathbf{E}[X^2]}{\mathbf{E}[X]^2}\le1+\frac1{\mathbf{E}[X]}+o(1)=1+\frac1{\delta}+o(1).
\end{equation}

Finally, by the Cauchy-Schwarz inequality,
$$\mathbf{Pr}(X>0)\ge\frac{\mathbf{E}[X]^2}{\mathbf{E}{[X^2]}}\ge\frac{\delta}{\delta+1}+o(1),$$
which completes the proof.
\end{proof}
Using the same method in \cite{zhou2026}, we obtain  
\begin{align}\label{eq:unique}
    \Pr(G_d(n,p)\text{ has a unique dominating set of size }\ln n)\ge \delta(1-\delta)/(1+\delta).
\end{align}
\subsection{Non-existence of  dominating sets of size $k$}
\begin{definition}
A subset $S\subset V$ is called a \textbf{quasi-dominating set} if $S$ dominates all but exactly one vertex of $G$; that is, there exists precisely one vertex in $V\setminus S$ that shares no common hyperedge with any vertex in $S$, while every other vertex in $V\backslash S$ is dominated by $S$.
\end{definition}
\begin{lemma}\label{lem:noset}
In  $G_d(n,p)$, if there is no dominating set of size $k=\ln n$, then with high probability there exists a quasi-dominating set.
\end{lemma}

The proof of Lemma \ref{lem:noset} is analogous to \cite{zhou2026}, and we move it to the Appendix.

\subsection{Proof of Theorem  \ref{th:main}}

Let $\mathscr{G}$ denote the family of instances of $G_d(n,p)$ such that each instance either has a unique dominating set of size $k=\ln n$ or has no dominating set of that size.  By Lemma~\ref{lem:onethird} and Lemma~\ref{lem:noset}, the probability
that a $d$-uniform hypergraph $G_d(n,p)$ belongs to $\mathscr{G}$ is bounded away from zero.

We show that for any constant $0<c<1$, no subgraph of order at most $n^c$ suffices
to decide whether $G$ contains a dominating set of size $k$.

\medskip
\noindent\textbf{Case 1: $G$ has a unique dominating set of size $k$.}

Assume that $G\in\mathscr{G}$ contains a unique dominating set $S$ of size $k$.
By (\ref{eq:unique}), this occurs with positive probability.
Let $H$ be an arbitrary induced subgraph of $G$ on at most $n^c$ vertices, and let $V_H$ be the vertex set of $H$.

Then
\[
\mathbf{P}(S\subseteq V\setminus V_H)
=\frac{\binom{n-n^c}{k}}{\binom{n}{k}}
=1-o(1),
\]
so with high probability the vertices of $S$ lie entirely outside $H$.

We first show that with high probability there exists a vertex in
$V\setminus(V_H\cup S)$ that is dominated by exactly one vertex in $S$ through exactly one hyperedge.
For a vertex $v\notin S$, define
\begin{align*}
A_v&=\{\text{$v$ is dominated by exactly one vertex of $S$ through exactly one hyperedge}\},\\
B_v&=\{\text{$v$ is dominated by $S$}\}.
\end{align*}
Then
\[
\mathbf{P}(A_v\mid B_v)
=\frac{\mathbf{P}(A_v)}{\mathbf{P}(B_v)}.
\]
Let $u$ be the unique vertex in $S$ that dominates $v$. The number of hyperedges that contain both $u$ and $v$ is $\binom{n-k-1}{d-2}$. Since $v$ is dominated by $u$ through exactly one hyperedge, precisely one of these hyperedges appears, which occurs with probability $\binom{n-k-1}{d-2}p(1-p)^{\binom{n-k-1}{d-2}-1}$. Moreover, $v$ shares no hyperedge with any other $u'\in S\backslash\{u\}$. The probability that none of the hyperedges connecting $v$ and any such $v'$ appears is $(1-p)^{(k-1)\binom{n-k-1}{d-2}}$. Finally, there must be no hyperedge that contains $v$ and at least two vertices from $S$. The total number of hyperedges containing $v$ is $\binom{n-1}{d-1}$. Among these, $\binom{n-1-k}{d-1}$ hyperedges  contain no vertex from $S$, and $k\binom{n-1-k}{d-2}$ hyperedges contain $v$ and exactly one vertex from $S$.  Hence, the number of hyperedges that contain $v$ and at least two vertices from $S$ is
$$M-k\binom{n-1-k}{d-2}.$$

The probability that none of these hyperedges appears is
$$(1-p)^{M-k\binom{n-1-k}{d-2}}.$$
Then
\begin{align*}
\Pr(A_v)=k\binom{n-1-k}{d-2}p(1-p)^{M-1}.
\end{align*}

Note that $(1-p)^M\approx\frac{\ln^2n}n$, $M\approx\frac{k}{(d-2)!}n^{d-2}$, and $p\approx\frac{(d-2)!}{n^{d-2}}.$ Therefore, $$\Pr(A_v)\approx\frac{\ln^3n}{n}.$$
We already know that 
\[\Pr(B_v)=1-(1-p)^{M}.\]
Consequently,
\[
\Pr\!\left(\exists\,v\in V\setminus(V_H\cup S): A_v\right)
=1-\left(1-\frac{\Pr(A_v)}{\Pr(B_v)}\right)^{n-n^c-k}
=1-o(1).
\]

Thus, with high probability there exist a vertex $v\in V\setminus(V_H\cup S)$ such that $v$ is dominated by a unique vertex $u\in S$ through exactly one hyperedge $e_{v,u}=(v,u,...)$. 
Choose a  vertex $v'\in V\setminus(V_H\cup S)$ and $u'\in S$ such that $e_{v',u'}=(v',u',...)$ is a hyperedge connecting them.

We now perform a \emph{symmetry mapping} by replacing the hyperedges $e_{v,u}$ and $e_{v',u'}$ by $e_{v,v'}=(v,v',...)$ and $e_{u,u'}=(u,u',...)$, as illustrated in Figure \ref{map} (transition from (a) to (b)). This operation modifies only hyperedges incident to vertices outside $H$, while preserving the degree of every vertex and keeping the total number of hyperedges unchanged. As a result of this transformation, vertex $v$ is no longer dominated by $S$, thus $S$ ceases to be a dominating set.

Moreover, with high probability no new dominating set of size $k$ is created.
Indeed, the probability that $\{v\}$ (or $\{v'\}$) extends to a
dominating set of size $k$ is at most
\[
\binom{n-n^c}{k-1}\bigl(1-(1-p)^{M'}\bigr)^{n-n^c-k-1}
=o(1),
\]
where $M'=\binom{n-1-n^c}{d-1}-\binom{n-k-1-n^c}{d-1}$.

Thus, with high probability, the modified graph contains no dominating set of
size $k$.
\begin{figure}[h]
  \centering
  \includegraphics[width=0.8\columnwidth]{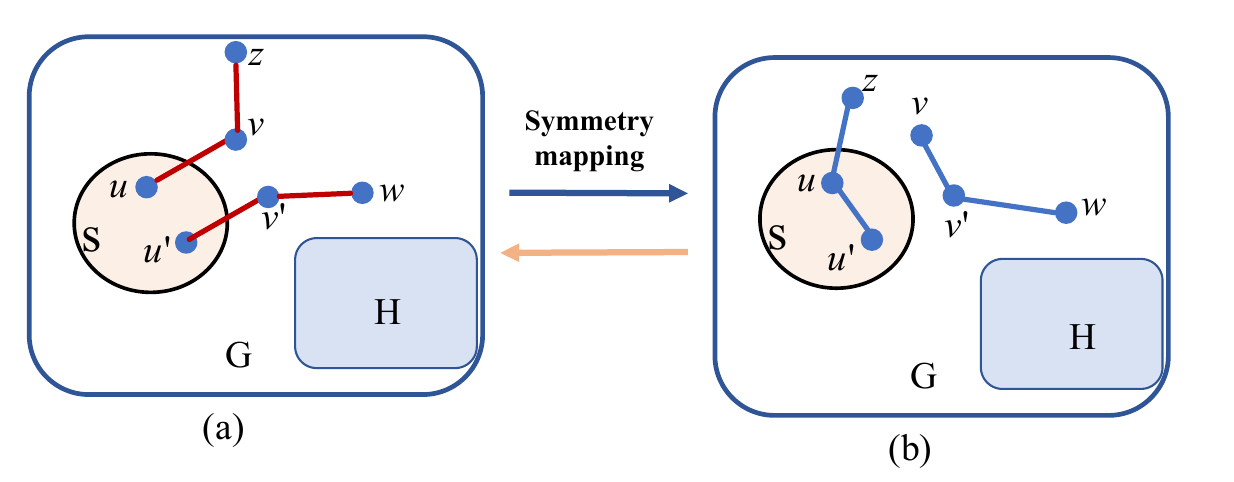}
\caption{\footnotesize{
A symmetry mapping between two classes of instances (for simiplicity, take $d=3$). \\
  (a) → (b) Initially, the vertex $v$ is dominated by exactly one vertex $u\in S$ through exactly one hyperedge. The symmetry mapping replaces the two hyperedges $e_{u,v}=(u,v,z),e_{u',v'}=(u',v',w)$  with $e_{u,u'}=(u,u',z),e_{v,v'}=(v,v',w)$. Consequently, after this transformation, $G$ has no dominating set of size $k$.\\
(b) → (a) Initially,  $v$ is the only vertex not dominated by $S$. The symmetry mapping transforms $G$ which has no dominating set of size $k$ into a graph with a unique dominating set of size $k$.
}\label{map}}
\end{figure}
\medskip

\noindent\textbf{Case 2: $G$ has no dominating set of size $k$.}

Now suppose that $G\in\mathscr{G}$ has no dominating set of size $k$.
By Lemma~\ref{lem:noset}, with high probability there exists a quasi-dominating set
$S$ ($|S|=k$) dominating all but one vertex. Let $v$ denote the unique undominated vertex.

Let $H$ be any induced subgraph of order at most $n^c$. Then
\[
\mathbf{P}(S\cup\{v\}\subseteq V\setminus V_H)
=\frac{\binom{n-n^c}{k+1}}{\binom{n}{k+1}}
=1-o(1),
\]
so with high probability all vertices in $S\cup\{v\}$ lie outside $H$.

Choose distinct vertices $u,u'\in S$ and $v'\in V\setminus(V_H\cup S)$ such that there exist hyperedges $e_{u,u'}=(u,u',...)$ and $e_{v,v'}=(v,v',...)$. Apply the symmetry mapping in the opposite direction by
replacing the hyperedges $e_{u,u'},e_{v,v'}$ with $e_{v,u}=(v,u,...)$ and $e_{v',u'}=(v',u')$. After this transformation, vertex $v$ becomes
dominated by $S$, and hence $S$ becomes a dominating set of size $k$. See Figure \ref{map}(from (b) to (a)).

As before, with high probability no other dominating set of size $k$ is created.
Therefore, with high probability, the modified graph contains a unique dominating
set of size $k$.

\medskip

In both cases, by altering only hyperedges whose vertices lie outside $H$, we can
flip the existence of a dominating set of size $k$ while keeping the induced
subgraph $H$ unchanged.  This yields self-referential instances for the dominating set problem. 
Consequently, for any $0<c<1$, no subgraph of order
$n^c$ contains sufficient information to determine whether $G_d(n,p)$ has a
dominating set of size $k$.

This proves that the dominating set problem for $G_d(n,p)$ is irreducible, and
completes the proof of Theorem~\ref{th:main}.

\section{Conclusions}

Solution independence  is the fundamental reason why exhaustive search becomes unavoidable. 
Self-reference and diagonalization provide a natural and powerful framework for establishing the necessity of exhaustive search. Indeed, Cantor constructed the classical diagonalization method to show that the set of real numbers has strictly larger cardinality than that of the rationals. The essence of this argument lies in the independence of coordinates in real number representations, which enables the construction of a new element that differs from every sequence in at least one position. This illustrates that solution independence serves as the structural foundation for constructing self-referential instances. Furthermore, by leveraging self-reference and diagonalization, one can rigorously prove that such self-referential instances are inherently indistinguishable by non-exhaustive algorithms. This logical chain establishes a theoretical framework for proving extreme hardness and characterizing the distinguishability boundary of non-exhaustive algorithms.

\appendix

\section{The existence of quasi-dominating sets}

\begin{proof}
Let $S$ be a $k$-vertex set, and define $$N=\sum_{|S|=k}I_S,$$ where $I_S=\mathbf{1}_{\{\text{$S$ is a quasi-dominating set}\}}$. For a fixed vertex $v\in V\setminus S$, we already know that the probability that $v$ is not dominated by $S$ is
\begin{align*}
q_0\triangleq(1-p)^M.
\end{align*}
The remaining $n-k-1$ vertices in $ V\setminus S\cup\{v\}$ must all be dominated by $S$.
A direct calculation gives
\begin{align}\label{eq:first2}
\mathbf{E}[N]=\binom{n}{k}(n-k)(1-p)^M\big(1-(1-p)^M\big)^{n-k-1}=\binom{n}{k}(n-k)q_0\big(1-q_0\big)^{n-k-1}.
\end{align}

Note that
\[
\frac{\mathbf{E}[N]}{\mathbf{E}[X]}
=(n-k)\frac{(1-p)^M}{1-(1-p)^M},
\]
and under our choice
\[
(1-p)^k\approx\frac{\ln^2 n}{n},
\]
we obtain
\[
\mathbf{E}[N]\approx(\ln n)^2\mathbf{E}[X]\to\infty.
\]
Thus the expected number of quasi-dominating sets tends to infinity. To prove the existence of quasi-dominating sets with high probability, we apply the second moment method. Let $S_1,S_2$ be two $k$-vertex set with $i=|S_1\cap S_2|$. Then
\[
\mathbf{E}[N^2]
=\sum_{|S_1|=k}\sum_{|S_2|=k}\mathbf{E}[I_{S_1} I_{S_2}]=\sum_{i=0}^k \Phi(i)W(i),
\]
where
\begin{align*}
\Phi(i)&=\binom{n}{k}\binom{k}{i}\binom{n-k}{k-i},\\
W(i)&=\mathbf{E}\left[I_{S_1} I_{S_2}\bigm| |S_1\cap S_2|=i\right]
= P_1(i)+P_2(i)+P_3(i)+2P_4(i),
\end{align*}
and 
\begin{align*}
P_1(i)&=m_iq_{00}(1-q_0)^{m_i-1}(1-q_0)^{2(k-i)}.\\
P_2(i)&=m_i(m_i-1)(q_0-q_{00})^2q_{11}^{m_i-2}(1-q_0)^{2(k-i)},\\
P_3(i)&=(k-i)^2q_0^2q_{11}^{m_i}(1-q_0)^{2(k-i)-2},\\
P_4(i)&=(k-i)m_iq_0(q_0-q_{00})q_{11}^{m_i-1}(1-q_0)^{2(k-i)-1}.
\end{align*}
with
$q_0\triangleq(1-p)^M,\ q_{00}\triangleq(1-p)^{M_i},\ q_{11}\triangleq1-2q_0+q_{00},\  m_i\triangleq n-2k+i.$

\medskip
Note that $$W(i)=\mathbf{P}\big(I_{S_1}=1,I_{S_2}=1\bigm| |S_1\cap S_2|=i\big),$$ where $I_{S_1}=1$ (respectively $I_{S_2}=1)$ denotes the event that $S_1$ (resp. $S_2$) dominates all but one vertex.
Define the disjoint vertex sets 
\begin{align*}
A=S_1\backslash S_2,B=S_2\backslash S_1,C=S_1\cap S_2,R=V\backslash(S_1\cup S_2).
\end{align*}
Then
\begin{align*}
|A|=|B|=k-i,|C|=i,|R|= n-2k+i.
\end{align*}
The events $I_S=1$ and  $I_{S'}=1$ can equivalently restated as follows:
\begin{itemize}
\item $I_{S_1}=1$: Among the vertices in $B\cup R$, exactly one vertex has no neighbor  in $S_1$;
\item $I_{S_2}=1$: Among the vertices in $A\cup R$, exactly one has no neighbor in $S_2$.
\end{itemize}

Let $x\in B\cup R$ be  the unique vertex not dominated by $S_1$, and $y\in A\cup R$ be the unique vertex not dominated by $S'$. We consider all possible locations of the pair $(x,y)$.

The probability that a vertex $u$ is not dominated by $S_1$ (or $S_2$) is $(1-p)^M$. The probability that  $u$ is not dominated by both $S_1$ and $S_2$ is $(1-p)^{M_i}$, since $M_i=\binom{n-1}{d-1}-\binom{n-1-2k+i}{d-1}$ is the number of hyperedges involving $S_1\cup S_2$. Let $E_1,E_2$ be the event that a vertex is dominated by $S_1,S_2$, respectively. Then
\[\Pr(\overline{E_1})=\Pr(\overline{E_2})=(1-p)^M,\quad\Pr( \overline{E_1}\cap\overline{E_2})=(1-p)^{M_i}.\]
Thus
\begin{align*}
 \Pr(u\text{ is not dominated by }S_1,u\text{ is dominated by }S_2) &=\Pr(\overline{E_1}\cap E_2)\\&=\Pr(\overline{E_1})-\Pr( \overline{E_1}\cap\overline{E_2})\\
 &=q_0-q_{00}.
 \end{align*}
and
 \begin{align*}
 \Pr(u\text{ is dominated by both }S_1\text{ and }S_2)
 =&\Pr(E_1\cap E_2)\\=&1-\Pr(\overline{E_1})-\Pr(\overline{E_2})+\Pr( \overline{E_1}\cap\overline{E_2})\\
 =&1-2q_0+q_{00}.
 \end{align*}

There are five possibilities for the locations of the pair $(x,y)$.

\textbf{Case 1: $x=y\in R$}.
There are $m_i$ choices for $x$. The vertex $x$ is not dominated by $S_1$ or $S_2$, which happens with probability $q_{00}$. Every $u\in R\backslash\{x\}$ must be dominated by both $S_1$ and $S_2$, with probability $q_{11}$. Every vertex $u\in A$ must be dominated by $S_2$, with probability $1-q_0$; similarly, every $u\in B$ should be dominated by $S_1$, also with probability $1-q_0$. Thus 

$$P_1(i)=m_iq_{00}q_1^{r-1}(1-q_0)^{2(k-i)}.$$

\textbf{Case 2: $x,y\in R,$ $x\ne y$}.
There are $m_i(m_i-1)$ choices for the ordered pair $(x,y)$. Note that $x$ is not dominated by $S_1$ but dominated by $S_2$, which happens with probability $q_0-q_{00}$. Similarly, $y$ is not dominated by $S_1$ but dominated by $S_2$, also with probability $q_0-q_{00}$. Every  $u\in R\backslash\{x,y\}$ must be dominated by both $S_1$ and $S_2$, with probability $q_{11}$. Every vertex $u\in A$ must be dominated by $S_2$, and every vertex $u\in B$ should be dominated by $S_1$, each with probability $1-q_0$. Hence

$$P_2(i)=m_i(m_i-1)(q_0-q_{00})^2q_{11}^{m_i-2}(1-q_0)^{2(k-i)}.$$

\textbf{Case 3: $x\in B,y\in A$}.
There are $k-i$ choices for $x$ and $k-i$ choices for $y$. Vertex $x$ is not dominated by $S_1$ with probability $q_0$; vertex $y$ is not dominated by $S_2$ with probability $q_0$. All vertices in $R$ are dominated by both $S_1$ and $S_2$, with probability $q_{11}^{m_i}$. Vertices in $A\backslash\{y\}$ must be dominated by $S_2$, and vertices in $B\backslash\{x\}$ must be dominated by $S_1$, each with probability is $(1-q_0)^{k-i-1}$.
We obtain

$$P_3(i)=(k-i)^2q_0^2q_{11}^{m_i}(1-q_0)^{2(k-i)}.$$

\textbf{Case 4: $x\in B,y\in R$}. (the symmetric case $x\in R,y\in A$ yields the same contribution).
 
 There are $k-i$ choices for $x$ and $m_i$ choices for $y$. Vertex $x$ is not dominated by $S_1$ with probability $q_0$; vertex $y$ is not dominated by $S_2$ but dominated by $S_1$, with probability $q_0-q_{00}$.
 Vertices in $R\backslash\{y\}$ are dominated by both $S_1$ and $S_2$, with probability $q_{11}^{m_i-1}$. Vertices in $A$ must be dominated by $S_2$, with probability $(1-q_0)^{k-i}$. Vertices in $B\backslash\{x\}$ must be dominated by $S_1$, with probability $(1-q_0)^{k-i-1}$.
Hence

$$P_4(i)=(k-i)m_iq_0(q_0-q_{00})q_{11}^{m_i-1}(1-q_0)^{2(k-i)-1}.$$

Combining the above cases, we have
$$W(i)=P_1(i)+P_2(i)+P_3(i)+2P_4(i).$$

\medskip
\medskip

Next, we estimate $\mathbf{E}[N^2]/\mathbf{E}[N]^2$.

\medskip
(1) If $i=0$, then $q_0\approx\frac{\ln^2n}{n-k}$, \ $q_{00}\approx q_0^2,\ q_{11}=1-o(1)$. Then
\begin{align*}
P_1(0)&=(n-2k)q_{00}(1-q_0)^{n-1}.\\
P_2(0)&=(n-2k)(n-2k-1)(q_0-q_{00})^2q_{11}^{n-2k-2}(1-q_0)^{2k},\\
P_3(0)&=k^2q_0^2q_{11}^{n-2k}(1-q_0)^{2k-2},\\
P_4(0)&=k(n-2k)q_0(q_0-q_{00})q_{11}^{n-2k-1}(1-q_0)^{2k-1}.
\end{align*}
A direct computation yields
\begin{align*}
\frac{W(0)}{(n-k)^2q_0^{2}(1-q_0)^{2n-2k-2}}=1+o(1).
\end{align*}
Note that
\begin{align*}
\Phi(0)=\binom{n}{k}\binom{n-k}{k}=(1+o(1))\binom{n}{k}^2.
\end{align*}

Hence
\[
\frac{\Phi(0)W(0)}{\mathbf{E}[N]^2}=1+o(1).
\]

\medskip
(2) If $i=k$, then the two quasi-dominating sets coincide, thus
\begin{align*}
\Phi(k)=\binom{n}{k},\quad
W(k)=(n-k)q_0(1-q_0)^{n-k-1},
\end{align*}
thus $\Phi(k)W(k)=\mathbf{E}[N]$.
Therefore,
\[
\frac{\Phi(k)W(k)}{\mathbf{E}[N]^2}
=\frac{1}{\mathbf{E}[N]}
=o(1).
\]

\medskip
(3) If $1\le i\le k-1$, it is easy to see that

\[
\frac{\Phi(i)}{\binom{n}{k}^2}\le\frac{k^{2i}}{n^i}.
\]

Moreover, note that $k=\ln n$. Standard asymptotic estimates give $$\frac{M}{M_i}\approx\frac{k}{2k-i},\quad q_0=(1-p)^M\approx \frac{\ln^2n}{n},$$ and consequently $$(1-p)^{M_i-2M}\approx \left(\frac{\ln^2n}n\right)^{i/k}.$$

\medskip

Therefore,
\begin{align*}
\frac{P_1(i)}{(n-k)^2q_0^2(1-q_0)^{2n-2k-2}}&=\frac{n-2k+i}{(n-k)^2}(1-p)^{M_i-2M}(1-q_0)^{-n+2k-i-1}\\
&=\Theta\left(\frac{1}{n}\right)\left(\frac{\ln^2n}n\right)^{i/k}\exp\{-\ln^2n\}.\\
\frac{P_2(i)}{(n-k)^2q_0^2(1-q_0)^{2n-2k-2}}&=\frac{(n-2k+i)(n-2k+i-1)}{(n-k)^2}\frac{(q_0-q_{00})^2}{q_0^2}\frac{(1-2q_0+q_{00})^{n-2k+i-2}}{(1-q_0)^{2n-4k+2i-2}}\\
&=\Theta\left(\exp\left\{\frac{(\ln^2n)^{2-\frac ik}}{n^{1-\frac ik}}\right\}\right).\\
\frac{P_3(i)}{(n-k)^2q_0^2(1-q_0)^{2n-2k-2}}&=\frac{(k-i)^2}{(n-k)^2}\frac{(1-2q_0+q_{00})^{n-2k+i}}{(1-q_0)^{2n-4k+2i}}\\
&=\Theta\left(\frac{k^2}{n^2}\exp\left\{\frac{(\ln^2n)^{2-\frac ik}}{n^{1-\frac ik}}\right\}\right).\\
\frac{P_4(i)}{(n-k)^2q_0^2(1-q_0)^{2n-2k-2}}&=\frac{(k-i)(n-2k+i)}{(n-k)^2}\left(1-\frac{q_{00}}{q_0}\right)\frac{(1-2q_0+q_{00})^{n-2k+i-1}}{(1-q_0)^{2n-4k+2i-1}}\\
&=\Theta\left(\frac kn\exp\left\{\frac{(\ln^2n)^{2-\frac ik}}{n^{1-\frac ik}}\right\}\right).
\end{align*}
Note that $k=\ln n$. Hence
\begin{align*}
\sum_{i=1}^{k-1}\frac{\Phi(i)W(i)}{\mathbf{E}^2[N]}=o(1).
\end{align*}
Therefore,
\begin{align*}
\frac{\mathbf{E}[N^2]}{\mathbf{E}^2[N]}\le1+o(1).
\end{align*}
The second moment method yields 
\begin{align*}
\mathbf{P}[N>0]\ge\frac{\mathbf{E}^2[N]}{\mathbf{E}[N^2]}\ge1-o(1).
\end{align*}
Thus with high probability there exists a quasi-dominating set that dominates all but exactly one vertex. 
\end{proof}

\end{document}